\documentclass[letterpaper, 10 pt, conference]{ieeeconf}  

\IEEEoverridecommandlockouts                              
\usepackage{amsmath}
\usepackage{amssymb}
\usepackage{graphicx}

\usepackage{xcolor}

\usepackage{tikz}
\usetikzlibrary{shapes,arrows,positioning}

\newcommand{\R}{\mathbb{R}} 
\newcommand{\Symm}{\mathbb{S}} 
\newcommand{\Ident}{I} 

\newcommand{\Zero}{0} 
\DeclareMathOperator{\diag}{diag} 
\DeclareMathOperator{\determinant}{det}

\newcommand{\Pb}{\mathbb{P}} 
\newcommand{\distr}{\sim} 
\newcommand{\Normal}{\mathcal{N}} 
\newcommand{\GP}{\mathcal{GP}} 

\newcommand{\bb}[1]{\mathbf{#1}} 
\newcommand{\bbm}[1]{\mathbf{#1}} 

\newtheorem{theorem}{Theorem}
\newtheorem{proposition}{Proposition}
\newtheorem{lemma}{Lemma}
\newtheorem{assumption}{Assumption}

\newcommand{\Ptrue}{P_\ast}
\newcommand{\Deltaset}{\mathbf{\Delta}} 
\newcommand{\RL}{\mathcal{RL}} 
\newcommand{\Lp}{\mathcal{L}} 
\newcommand{\starprod}{\star} 

\title{\LARGE \bf
Learning-enhanced robust controller synthesis with rigorous statistical and control-theoretic guarantees
}

\author{Christian Fiedler\textsuperscript{\rm 1,2}, Carsten W. Scherer\textsuperscript{\rm 3} and Sebastian Trimpe\textsuperscript{\rm 1,2}
\thanks{Funded by Deutsche Forschungsgemeinschaft (DFG, German Research Foundation) under Germany’s Excellence Strategy - EXC 2075 – 390740016 and in part by the Cyber Valley Initiative. We acknowledge the support by the Stuttgart Center for Simulation Science (SimTech).}
\thanks{\textsuperscript{\rm 1}Institute for Data Science in Mechanical Engineering, RWTH Aachen University, Aachen, Germany {\tt\small \{christian.fiedler,trimpe\}@dsme.rwth-aachen.de}}
\thanks{\textsuperscript{\rm 2}Intelligent Control Systems Group, Max Planck Institute for Intelligent Systems, Stuttgart, Germany}
\thanks{\textsuperscript{\rm 3}Department of Mathematics, University of Stuttgart, Stuttgart, Germany
        {\tt\small carsten.scherer@mathematik.uni-stuttgart.de}  
}%
}

\begin{document}

\maketitle
\thispagestyle{empty}
\pagestyle{empty}

\begin{abstract}
The combination of machine learning with control offers many opportunities, in particular for robust control. However, due to strong safety and reliability requirements in many real-world applications, providing rigorous statistical and control-theoretic guarantees is of utmost importance, yet difficult to achieve for learning-based control schemes. We present a general framework for learning-enhanced robust control that allows for systematic integration of prior engineering knowledge, is fully compatible with modern robust control and still comes with rigorous and practically meaningful guarantees. Building on the established Linear Fractional Representation and Integral Quadratic Constraints framework, we integrate Gaussian Process Regression as a learning component and state-of-the-art robust controller synthesis. 
In a concrete robust control example, our approach is demonstrated to yield improved performance with more data, while guarantees are maintained throughout.
\end{abstract}

\section{INTRODUCTION}
Many methods of modern control rely on accurate plant models \cite{dullerudpaganini}. Traditionally, models are obtained from first-principles modeling \cite{astrommurray}, harnessing extensive expert domain knowledge, or are derived experimentally using system identification \cite{ljung_sys_ident}. In most cases, a combination of these strategies is used: First-principles models are enhanced with components derived from data. 
Therefore, recent advances in machine learning offer tremendous opportunities for control. By using advanced learning approaches, it is possible to improve models from real-world data sets, even in established areas of control engineering, cf. e.g. \cite{pillonettoetal_kernel_sys_ident_review}. 
Furthermore, modern control systems are increasingly complex and learning-based approaches offer the chance to tackle this complexity.

However, control applications pose significant challenges for learning approaches, such as potential non-independence of data-samples, unmeasured states and real-time feasibility issues.
Furthermore, rigorous guarantees on the behavior of closed-loop control systems are of paramount importance in many applications, in particular, in safety-critical areas like aerospace control systems and human-robot interaction scenarios. Using model-based approaches, rigorous and practically relevant guarantees can often be given, mostly in the form of stability and constraint satisfaction guarantees. 
However, including learning-based components in control systems significantly complicates the situation, and giving theoretical guarantees on the overall system can be challenging. 
Additionally, for practical applicability of learning-based control it is important to leverage existing prior knowledge in a systematic manner. In many real-world scenarios, extensive domain and expert knowledge is available, often in the form of first-principles modeling utilizing disciplines like physics, chemistry or biology, as well as vast amounts of engineering experience and intuition. In order to make learning-based control approaches real-world feasible, reliable and data efficient, it is important to harness this prior knowledge.

In this paper we present a general methodology bringing machine learning and modern robust control closer together.
We utilize Gaussian Process Regression (GPR) \cite{rasmussen_williams_gp} together with rigorous, yet practical frequentist uncertainty bounds to learn unknown components of a system. In particular, we demonstrate how to use prior knowledge in a systematic manner in the learning process. Based on the statistical uncertainty bounds we then derive established system-theoretic uncertainty descriptions. We advocate for the use of the Linear Fractional Representation (LFR) \cite{zhougloverdoyle,scherer_theory_rc} and Integral Quadratic Constraint (IQC) \cite{megretskirantzer_iqc,veenmanetal_iqc_review}  framework to integrate the learning components with modern robust control approaches, in particular, robust controller synthesis with control-theoretic guarantees like robust stability and performance.
In this way, we can give rigorous guarantees and can harness established engineering prior knowledge, cf. Figure \ref{fig.lfr_diagram} for an illustration.
\begin{figure}[h!]
\centering
\tikzstyle{signal} = []
\tikzstyle{block} = [draw, rectangle, 
    minimum height=3em, minimum width=3em]
\begin{tikzpicture}[auto, >=latex']
    \node [block, name=input, minimum height=1cm, minimum width=1cm] (P) {$
    \begin{matrix}
	\begin{array}{c | c }
	A 		& B \\
	\hline
	C & D
	\end{array}
	\end{matrix}
    $};
    \node[signal,name=w,left=2cm of P] {$w$};
    \node[signal,name=z,right=2cm of P] {$z$};

    \node[block,below=0.4cm of P] (K) {$K$};
	 \node[matrix,draw,above=0.3cm of P] (Delta)  {
    			\node{\includegraphics[width=1cm]{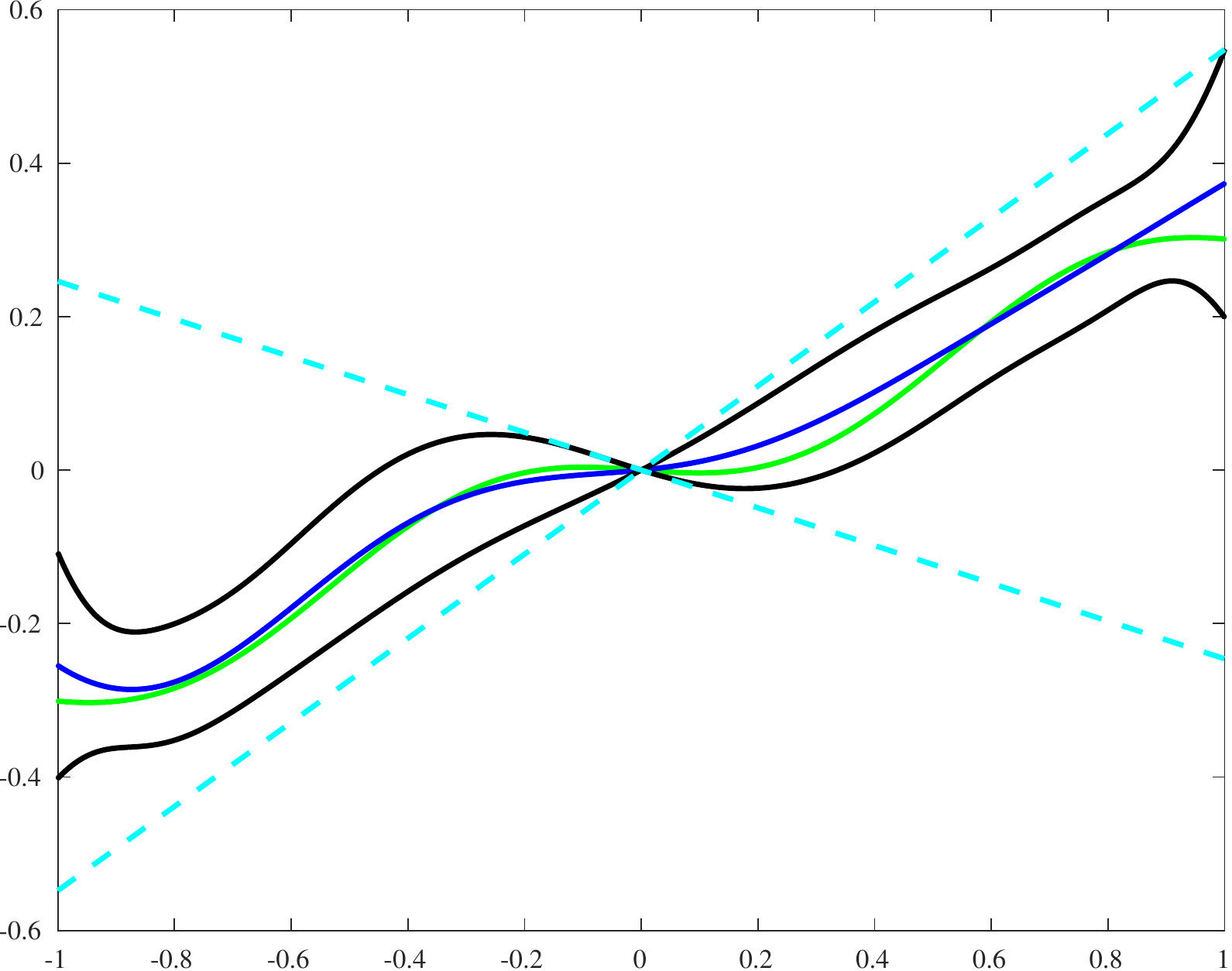}}; & ; \\
    			& \node{\includegraphics[width=1cm]{media/sector_bounds_plain}}; \\
    		}; 
  	\node[text opacity=0.4] at (Delta) {\Huge $\Delta$};
  	
    \draw [->] ([yshift=-0.2cm] P.east) -- node[pos=0.5,auto=right] {$y$} ++(1cm,0) |- (K.east);
    \draw [->] (K.west) -- ++(-1.5cm,0) |- node[pos=0.8,auto=right] {$u$} ([yshift=-0.2cm]P.west);
    \draw [->] (w) --  (P);
    \draw [->] (P) -- (z);
    \draw [->] ([yshift=0.2cm] P.east) -- node[pos=0.5,auto=left] {$q$} ++(1cm,0) |- (Delta.east);
    \draw [->] (Delta.west) -- ++(-0.5cm,0) |- node[pos=0.8,auto=left] {$p$} ([yshift=0.2cm] P.west);
    
\end{tikzpicture}

\caption{Illustration of the LFR framework in the context of learning-enhanced robust controller synthesis.
All uncertain and trouble-making parts are "pulled out" into the $\Delta$-block.
We use the framework to combine prior knowledge with components learned from data (in the $\Delta$-block) in a systematic manner.}
\label{fig.lfr_diagram}
\end{figure}

\subsection{Related work}
Combining uncertainty quantification for machine learning with robust control approaches is a common strategy in learning based control, in particular, in learning-based Model Predictive Control \cite{hewingetal_lbmpc}. 
Using frequentist uncertainty results for GPR in a robust control setting has been used for example in 
\cite{kolleretal_lbmpc,umlauftetal_clf_gp,helwaetal_rlbc_tracking,devonportetal_bayesian_learning_control}, and the concept of $\delta$-safety from \cite{kolleretal_lbmpc} is conceptually very similar to our control-theoretic guarantees, cf. Theorem \ref{thm.controller_guarantees} below.
However, due to the difficult applicability of earlier uncertainty bounds for GPR,
only heuristics for the uncertainty sets are employed and, hence, all guarantees are lost in the end.
Furthermore, often considerable prior knowledge about the system is not used, in particular, the fine structure of the uncertainties is not utilized in the control schemes.
The LFR framework has been sporadically used in the context of learning-based control, cf. e.g. 
\cite{berkenkampschoellig_lbrc,berkenkampschoellig_safe_lbc_gp}. However, its versatility and modularity has not been taken advantage of before. In particular, to the best of our knowledge it has not been used in the context of control-theoretic guarantees for learning-based control schemes.
The recent works \cite{berberichetal_robust_datadriven_acc20,berberichetal_combining_prior_data} explore the usage of the LFR framework and modern robust controller synthesis in a learning context, however, they rely on a data-driven approach  \cite{houwang_datadriven_control_survey} and do not consider nonlinearities. As such these works can be seen as complementary to this paper.

Another alternative approach to learning-based controller synthesis uses a Bayesian framework, cf. e.g. \cite{vonrohretal_probabilistic_robust_lqr}, and can also be seen as complementary to the frequentist perspective taken in this work.
Furthermore, since our learning approach leads to uncertainty sets that can be interpreted as estimates of certain system-theoretic properties, our work is also related to recent methods for inferring such properties directly from data, cf. e.g. \cite{kochetal_verifying_dissipativity_data_noise}.

\subsection{Outline}
In Section \ref{sec.background} we provide the necessary background on GPR, the Reproducing Kernel Hilbert Space (RKHS) framework and frequentist uncertainty bounds for GPR as well as the on the LFR and IQC frameworks. 
In Section \ref{sec.method} a detailed description of our approach, its assumptions as well as possible extensions are given.
We demonstrate the methodology in Section \ref{sec.example} on a concrete example from robust control.
Finally, Section \ref{sec.conclusion} provides some concluding discussions, a summary of our work and describes ongoing and future work.

\subsection{Notation} \label{sec.notation}
We indicate that a symmetric matrix $A\in \Symm^{n \times n}$ is positive (semi)-definite by $A (\succeq) \succ 0$. $\Ident$ is the identity matrix. We define $[N]:=\{1,\ldots,N\}$. 
Real-rational proper matrix functions without poles on the extended imaginary axis are denoted by $\RL^{n \times n}_\infty$, 
$\Lp_2$ is the usual Lebesgue space and $\langle \cdot, \cdot \rangle$ the corresponding inner product.
The dimension of a signal $x$ is denoted by $n_x$.
We write
$\begin{bmatrix}
	\begin{array}{c | c }
	A 		& B \\
	\hline
	C & D
	\end{array}
\end{bmatrix}$
for the transfer matrix of an LTI system and $A \starprod B$ for the usual (lower) LFT of two compatible systems.
\section{BACKGROUND} \label{sec.background}
\subsection{Gaussian Process Regression} \label{sec.gpr}
Here we recall some basics about Gaussian Process Regression (GPR). For details we refer to \cite{rasmussen_williams_gp}.
For simplicity we restrict ourselves to the scalar setting, while the generalization to the multivariate setting is straightforward.
Let $X$ be a non-empty set. A function $k: X \times X \rightarrow \R$ is called a kernel (on $X$) if it is symmetric, i.e., $\forall x,x^\prime \in X$ $k(x,x^\prime)=k(x^\prime,x)$, and positive (semi)definite, i.e., $(k(x_i,x_j))_{i,j\in[N]} \succeq 0$ for all $x_1,\ldots,x_N \in X$.
Let $k$ be a kernel and $m: X \rightarrow \R$ an arbitrary function. A Gaussian Process (GP) with mean $m$ and covariance function $k$ is a collection of random variables $(f(x))_{x \in X}$ such that
$\begin{pmatrix} f(x_1) & \cdots & f(x_N) \end{pmatrix} \distr \Normal((m(x_i))_{i \in [N]}, (k(x_i,x_j))_{i,j\in[N]})$
for all $x_1,\ldots,x_N \in X$. Intuitively, such a GP, henceforth denoted by $f \distr \GP_X(m,k)$, is just a distribution over functions.
GPR is a Bayesian nonparametric regression method: A prior (distribution) is updated using data and a likelihood model to arrive at the posterior (distribution). In the case of GPR, the prior is given by a GP $f \distr \GP_X(m,k)$, the data and likelihood model are of the form $\mathcal{D}=((x_n,y_n)_{n \in [N]}$ and $y_n=f(x_n)+\epsilon_n$, respectively, with $\epsilon_n \overset{\text{i.i.d.}}{\distr} \Normal(0, \sigma_\epsilon^2)$ and the resulting posterior is again a GP, $f\lvert_{\mathcal{D}} \distr \GP_X(m_\mathcal{D}, k_\mathcal{D})$.
Without loss of generality we assume in the following that $m \equiv 0$, and in this case we have 
$m_\mathcal{D}(x)=\bb{k}_\mathcal{D}(x)^\top(\bbm{K}_{\mathcal{D}}+\sigma_\epsilon^2\Ident_N)^{-1}\bb{y}$ 
and 
$k_{\mathcal{D}}(x,x^\prime)=k(x,x^\prime)-\bb{k}_\mathcal{D}(x)^\top(\bbm{K}_{\mathcal{D}}+\sigma_\epsilon^2\Ident_N)^{-1}\bb{k}_\mathcal{D}(x^\prime)$ for $x,x^\prime \in X$, where we defined $\bb{k}_\mathcal{D}(x):=(k(x,x_n))_{n \in [N]}$ and the kernel matrix $\bbm{K}_{\mathcal{D}}:=(k(x_i,x_j))_{i,j\in [N]}$.
In practice the GPR model contains unspecified parameters, called hyperparameters, that also have to be learned from data. Usually these are estimated with likelihood optimization, for details we refer again to \cite{rasmussen_williams_gp}.
Furthermore, due to the inversion of an $N \times N$-matrix (where $N$ is the number of samples in the data set), approximate methods like sparse GPs have to be utilized in practice when dealing with large data sets, cf. e.g. \cite{baueretal_understanding_sparse_gp}.
Finally, it is possible to systematically include prior knowledge into GPR by choosing an appropriate covariance function \cite{rasmussen_williams_gp, duvenaud_automatic_model_construction}. 
\subsection{Frequentist uncertainty bounds for GPR}
In the robust control framework a fixed ground truth is assumed and we want an uncertainty set that contains the ground truth with certainty. In the stochastic setting of machine learning this requirement is typically relaxed: The uncertainty set should contain the ground truth with (very high) prespecified probability, the latter with respect to the data generating distribution. This is commonly referred to as the frequentist setting \cite{murphy_ml}.
The posterior variance from GPR can be used as a measure of the remaining uncertainty. Unfortunately, since GPR is a Bayesian method, the posterior distribution cannot directly be used to arrive at a frequentist uncertainty set \cite{fiedleretal_gp_bounds}. However, using the framework of Reproducing Kernel Hilbert Spaces (RKHS), it is possible to provide a link, as first demonstrated by \cite{srinivasetal_10}. 
Let $X$ be a non-empty set and $H \subseteq \R^X$ a Hilbert space. A function $k: X \times X \rightarrow \R$ is called a reproducing kernel (for $H$) if $\forall x \in X$ $x^\prime \mapsto k(\cdot,x) \in H$ and $\forall x \in X$, $f \in H$ $f(x)=\langle f, k(\cdot, x) \rangle$. A reproducing kernel is necessarily a kernel and every kernel $k$ has a unique RKHS  \cite{steinwartchristmann_svm_book}, denoted in the following by $H_k$, with associated norm $\|\cdot\|_k$.
If the ground truth is contained in the RKHS belonging to the covariance function used in GPR, we can derive rigorous uncertainty sets. The following result is a variant of \cite[Theorem~1]{fiedleretal_gp_bounds}, which in turn is based on \cite{cg17}.
\begin{proposition} \label{prop.uncertainty_bound}
Let $X$ be a non-empty set and $k: X \times X \rightarrow \R$ a kernel with RKHS $H_k$. Let $f \in H_k$ with $\|f\|_k \leq B$, $x_1,\ldots,x_N \in X$ and $y_n=f(x_n)+\epsilon_n$ for $n \in [N]$, where $\epsilon_1,\ldots,\epsilon_N$ are independent $R$-subgaussian random variables. Denote by $\mu_{\mathcal{D}}, \sigma_{\mathcal{D}}^2$ the posterior mean and variance, respectively, of GPR with prior $\GP_X(0,k)$, variance $\lambda>0$ in the likelihood model and using the data set $((x_n,y_n))_{n \in [N]}$. Let $\delta \in (0,1)$ and define
$\beta_{\mathcal{D}} := B + 2R\sqrt{\log\determinant(\bbm{K}_{\mathcal{D}}+\bar{\lambda}\Ident_N)-2\log(\delta)}$, where $\bar{\lambda}=\max\{1,\lambda\}$. Then
\begin{equation*}
\Pb[|f(x)-\mu_{\mathcal{D}}(x)| \leq \beta_{\mathcal{D}} \sigma_{\mathcal{D}}(x) \: \forall x \in X] \geq 1-\delta.
\end{equation*}
\end{proposition}
Note that once a reasonable upper bound $B$ on $\|f\|_k$ is available, the resulting uncertainty set is explicitly computable. 
Numerical experiments indicate that Proposition \ref{prop.uncertainty_bound} it is rather tight. Similar to previous works \cite{srinivasetal_10,cg17}, Proposition \ref{prop.uncertainty_bound} requires the usage of the same kernel as reproducing kernel and covariance function in GPR.
First steps towards results robust to kernel misspecifications have been taken in \cite{fiedleretal_gp_bounds}. However, for simplicity we do not discuss these issues in the following.
Furthermore, we would like to emphasize that no additional assumptions on the kernel are necessary in Proposition \ref{prop.uncertainty_bound}. In particular, a lot of flexibility is provided to include prior knowledge via the kernel and systematic approaches like \cite{jidlingetal_linearly_constrained_gp,geisttrimpe_gpgp} could be used.

\subsection{Modeling uncertainty using IQCs and LFR} \label{sec.lfr_iqc}
In real-world applications, only parts of the plant are unknown and even the remaining uncertainties often carry a lot of structure.
It is therefore important to be able to systematically combine known and unknown components of dynamical systems and to leverage sophisticated descriptions of the fine properties of uncertainties. An established framework in modern robust control, that is ideally suited for this task are Linear Fractional Representations (LFR). 
Consider the partially known system (using operator notation)
\begin{equation} \label{eq.true_sys}
\begin{pmatrix}
z \\ y
\end{pmatrix}
=
\Ptrue
\begin{pmatrix}
w \\ u
\end{pmatrix},
\end{equation}
with control input $u$, measured output $y$, generalized disturbance $w$ and controlled output $z$. In the LFR framework we model the known parts by a nominal LTI system
\begin{equation}
\begin{pmatrix} \label{eq.generalized_plant}
q \\
z \\
y
\end{pmatrix}
=
\begin{pmatrix}
	G_{qp} & G_{qw} & G_{qu} \\
	G_{zp} & G_{zw} & G_{zu} \\
	G_{yp} & G_{yw} & G_{zu}
\end{pmatrix}
\begin{pmatrix}
p \\
w \\
u
\end{pmatrix}
\end{equation}
in feedback connection with an uncertain system
\begin{equation} \label{eq.uncertain_sys}
p = \Delta(q),
\end{equation}
where $\Delta \in \Deltaset$, with the latter being the class of given uncertainties.
In many cases substantial additional information about $\Delta$ is known, in particular, if machine learning methods are employed to reduce the epistemic uncertainty about the system. Here we propose to use the versatile and powerful framework of Integral Quadratic Constraints (IQCs) for this task. 
We say that an uncertainty $\Delta$ fulfills the IQC defined by the multiplier $\Pi \in \RL^{(n_q+n_p)\times(n_q+n_p)}_\infty$
if 
\begin{equation} \label{eq.iqc_def}
\left\langle
\begin{pmatrix}
q \\
\Delta(q)
\end{pmatrix},
\Pi
\begin{pmatrix}
q \\
\Delta(q)
\end{pmatrix}
\right\rangle
\geq 0
\qquad \forall q \in \Lp_2.
\end{equation}
For more details and a very comprehensive collection of multipliers we refer to \cite{veenmanetal_iqc_review}.
Furthermore, IQCs can be used for robust stability and performance analysis:
Consider the system
\begin{equation} \label{eq.uncertain_cl}
\begin{pmatrix}
q \\
z
\end{pmatrix}
=
\begin{pmatrix}
\mathcal{G}_{qp} & \mathcal{G}_{qw} \\
\mathcal{G}_{zp} & \mathcal{G}_{zw}
\end{pmatrix}
\begin{pmatrix}
p \\
w
\end{pmatrix}
\end{equation}
resulting from connecting system \eqref{eq.generalized_plant} with a given LTI controller $K$.
If an IQC description of the uncertainty is available, the well-known IQC stability theorem from \cite{megretskirantzer_iqc} allows to conclude robust stability and even robust performance, cf. \cite{veenmanetal_iqc_review} for details.
Using the KYP Lemma
leads to an LMI problem and makes the robust stability and performance analysis with IQCs computationally tractable, see again \cite{veenmanetal_iqc_review} for details. 

\section{METHODOLOGY} \label{sec.method}
\subsection{Problem setting}
We propose to utilize the LFR framework described in Section \ref{sec.lfr_iqc} since it is an established and flexible framework for uncertain system modeling \cite{zhougloverdoyle,scherer_theory_rc}. In particular, it is well suited to integrate machine learning components into modern robust control: extensive prior system knowledge can be included into the nominal plant and all learning components are "pulled out" into the uncertain part. By using machine learning methods incorporating prior knowledge, it is possible to arrive at very data-efficient learning-enhanced control schemes. Furthermore, as we will demonstrate in the following, the LFR framework allows us to easily transfer statistical bounds into control-theoretic guarantees in a structured fashion.

For the rest of this section, we are concerned with an unknown plant \eqref{eq.true_sys} given in the LFR form \eqref{eq.generalized_plant}, \eqref{eq.uncertain_sys}. 
The overall goal is to perform robust controller synthesis and to give robust stability and performance guarantees for the synthesized controller. For sake of concreteness, we assume a single objective optimization where possible performance weighting filters have been already included in the generalized plant $P$ and we are left with a standard robust (against the uncertainty in \eqref{eq.uncertain_sys}) $H_\infty$ synthesis problem. 

Assume now that we are not satisfied with the achievable performance of the controller due to the uncertainty being too large. 
For this purpose, we will employ machine learning in order to reduce the epistemic uncertainty contained in \eqref{eq.uncertain_sys}, using a data set $\mathcal{D}$.
Since the LFR framework is very modular, it causes no limitation to focus our attention in this paper to a single uncertain component that is to be learned. 
We would like to stress that it is easily possible to deal with multiple uncertainties containing learned components, even using different learning methods, and to combine these  with uncertainties that classically emerge if capturing parametric model variations or unmodeled dynamics as resulting from a complexity reduction step, cf. \cite{dullerudpaganini}.

Instead of describing an abstract general procedure, we opt to present the concrete case of a static (partially) unknown diagonal nonlinearity,
\begin{equation} \label{eq.nl}
p=\phi(q),
\end{equation}
with $n_p=n_q$ and $p_i(t)=\phi_i(q_i(t))$ for $i\in[n_p]$, $t \geq 0$. We will comment on possible generalizations and resulting complications during the description of our approach.

\subsection{Learning static nonlinearities} \label{sec.learning_static_nl}
Since the described procedure can be applied to all $\phi_i$, $i \in [n_p]$ separately, we focus on a single scalar nonlinearity $\phi: \R \rightarrow \R$ in this and the next section.
The corresponding data set $\mathcal{D}=((x_n,y_n))_{n \in [N]}$ is of the form
\begin{equation} \label{eq.data_model}
y_n = \phi(x_n) + \epsilon_n,
\end{equation}
where $\epsilon_1,\ldots,\epsilon_N$ is additive measurement noise. 
This means that the unknown part can be isolated and input-output samples can be collected.
For concreteness, we make the following assumption on the noise.
\begin{assumption} \label{assumption.noise}
$\epsilon_1,\ldots,\epsilon_N$ are independent, $R$-subgaussian random variables with $R \geq 0$ the subgaussianity constant.
\end{assumption}

We propose to use GPR to learn from $\mathcal{D}$ and to use Proposition \ref{prop.uncertainty_bound} to arrive at a high probability uncertainty set. 
For the latter, we make the following assumption, as is commonly done in learning-based control using GPs \cite{kolleretal_lbmpc}.
\begin{assumption} \label{assumption.rkhs}
$k$ is a kernel on $\R$ and $\phi \in H_k$ with $\|\phi\|_k \leq B$.
\end{assumption}
This leads immediately to the following Lemma.
\begin{lemma} \label{lemma.uncertainty_set_gpr}
If applying GPR with covariance function $k$ and nominal noise level $\lambda > 0$ to data set $\mathcal{D}$, we get with the notation from Proposition \ref{prop.uncertainty_bound}, for any $\delta \in (0,1)$ under Assumption \ref{assumption.noise} and \ref{assumption.rkhs} that
\begin{equation*}
\Pb[|\phi(x)-\mu_{\mathcal{D}}(x)| \leq \beta_{\mathcal{D}} \sigma_{\mathcal{D}}(x) \: \forall x \in \R] \geq 1-\delta.
\end{equation*}
\end{lemma}

This is illustrated in Figure \ref{fig.sector_bounds} with the example from Section \ref{sec.example}.
Additional prior knowledge on the static nonlinearity can be systematically included in GPR via the covariance function. 

As a concrete example, suppose we know that $\phi(0)=0$, which is a requirement for a sector bounded nonlinearity. Given any kernel $k$ with $k(0,0)\not = 0$, following the procedure described in \cite{jidlingetal_linearly_constrained_gp} leads to the new kernel
\begin{equation} \label{eq.kernel0}
k_0(x,x^\prime) = k(x,x^\prime) - \frac{1}{k(0,0)}k(x,0)k(x^\prime,0).
\end{equation}
Note that in this particular example, the same effect can be achieved conditioning a GP prior on the noise free virtual data point $(0,0)$.
All functions contained in $H_{k_0}$ as well as all samples from $\GP(0,k_0)$ are guaranteed to be zero in zero. Many other properties of functions can be encoded in this way, cf. e.g. \cite{jidlingetal_linearly_constrained_gp,geisttrimpe_gpgp}.

\noindent
\emph{Remarks} 
\begin{enumerate}
\item It is possible to deal with multivariate static nonlinearities with multi-output GPs, cf. e.g. 
	\cite{rasmussen_williams_gp}.
\item The independence assumption on the noise is reasonable in the setting used here, since we directly collect input-output samples. Note that Proposition \ref{prop.uncertainty_bound} is based on the powerful concentration inequality \cite[Theorem~1]{cg17} which supports also much more general noise distributions. In particular, we expect that results like Proposition \ref{prop.uncertainty_bound} can be generalized to the closed-loop setting with process noise.
\item  Other machine learning methods can be used for this step, as long as numerical uncertainty bounds can be derived. For example, if we have magnitude bounds on the additive noise in \eqref{eq.data_model} and no distributional assumptions, then the kernel methods and accompanying uncertainty bounds in \cite{maddalenaetal_deterministic_uncertainty_bounds} or the Nonlinear Set Membership framework \cite{milanesenovara_nsm} could as well be employed.
\end{enumerate}
\subsection{Connection to robust control: High probability sector bounds} \label{sec.sector_bounds_from_gpr}
As common in modern robust control, the uncertainty set has to be transformed into a form that is amenable to controller synthesis \cite{dullerudpaganini}. As a concrete example, we illustrate this by deriving the tightest sector bounds compatible with the high probability uncertainty set from Section \ref{sec.learning_static_nl}. Recall that $\varphi: \R \rightarrow \R$ belongs to a sector $[\kappa_1,\kappa_2]$ if
\begin{equation} \label{eq.sector_condition}
\kappa_1 x^2 \leq x \varphi(x) \leq \kappa_2 x^2 \qquad \forall x \in \R.
\end{equation}
Generalizations to the multivariate settings are straightforward, cf. e.g. \cite{khalil_nls}.
Lemma \ref{lemma.uncertainty_set_gpr} leads immediately to the next result.
\begin{lemma} \label{lemma.sector_bounds_from_gpr}
Consider the situation of Lemma \ref{lemma.uncertainty_set_gpr} and assume that $\phi$ belongs to the sector $[\kappa_1,\kappa_2]$ with $\kappa_1 \leq \kappa_2$. If defining 
{ \scriptsize
\begin{align*}
\hat{\kappa}_1 & := \min_{x \in [a,b]\setminus \{ 0\}} \frac{\min\{x \cdot (\mu_{\mathcal{D}}(x) - \beta_{\mathcal{D}} \sigma_{\mathcal{D}}(x)), x \cdot (\mu_{\mathcal{D}}(x) + \beta_{\mathcal{D}} \sigma_{\mathcal{D}}(x)) \}}{x^2} \\
\hat{\kappa}_2 & := \max_{x \in [a,b]\setminus \{ 0\}} \frac{\max\{x \cdot (\mu_{\mathcal{D}}(x) - \beta_{\mathcal{D}} \sigma_{\mathcal{D}}(x)), x \cdot (\mu_{\mathcal{D}}(x) + \beta_{\mathcal{D}} \sigma_{\mathcal{D}}(x)) \}}{x^2}, 
\end{align*}  }
then $[\hat{\kappa}_1,\hat{\kappa}_2]$ is an overapproximation of $[\kappa_1,\kappa_2]$, i.e. $\hat{\kappa}_1 \leq \kappa_1$ and $\hat{\kappa_2} \geq \kappa_2$, with probability at least $1-\delta$.
\end{lemma}
This approach can be made rigorously computational for Lipschitz continuous functions by evaluating the uncertainty set from Lemma \ref{lemma.uncertainty_set_gpr} on a fine grid and adapting Lemma \ref{lemma.sector_bounds_from_gpr} correspondingly.
\subsection{Robust controller synthesis} \label{sec.robust_synthesis}
We now tackle the robust controller synthesis problem with the learned uncertainty component using
the versatile approach from \cite{veenmanscherer_iqc_synthesis} for the case of static IQC multipliers. 
Consider without loss of generality the following minimal realization of \eqref{eq.generalized_plant}
(with possible performance weights already included in the generalized plant)
\begin{equation}
\begin{pmatrix} \label{eq.generalized_plant_realization}
\dot{x} \\
\hline
q \\
z \\
y
\end{pmatrix}
=
\begin{pmatrix}
	\begin{array}{c | c  c  c}
	A 		& B_p 		& B_w 		& B_u \\
	\hline
	C_q 	& D_{qp} 	& D_{qw}	& D_{qu} \\
	C_z 	& D_{zp}	& D_{zw}	& D_{zu} \\
	C_y	& D_{yp}	& D_{yw}	& 0
	\end{array}
\end{pmatrix}
\begin{pmatrix}
x \\
p \\
w \\
u
\end{pmatrix}.
\end{equation}
Canceling the uncertainty channel $p \rightarrow q$ results in the system $G_0$.
A standard $H_\infty$ controller synthesis leads to an intial LTI controller $K_0$
with nominal performance level $\gamma_0$.
Connecting $K_0$ to $G$ on the control channel $u \rightarrow y$ results in 
\begin{equation} \label{eq.closed_loop}
 \mathcal{G}_0=G \starprod K_0=
\begin{bmatrix}
\begin{array}{c|c c}
\mathcal{A} 	& \mathcal{B}_p 	& \mathcal{B}_w \\
\hline
\mathcal{C}_q 	& \mathcal{D}_{qp}	& \mathcal{D}_{qw} \\
\mathcal{C}_z 	& \mathcal{D}_{zp} 	& \mathcal{D}_{zw}
\end{array}
\end{bmatrix}.
\end{equation}
The corresponding realization can be derived by elementary manipulations and is omitted due to space constraints. 
Assume now that $\Deltaset$ fulfills an IQC with multipliers $\Pi=P \in \mathcal{P}$, where $\mathcal{P}$ has an LMI description. 
This leads to an initial analysis LMI \eqref{eq.iqc_lmi}
\begin{figure*}[h!]
{\scriptsize
\begin{equation} \label{eq.iqc_lmi}
\begin{pmatrix}
\mathcal{X}\mathcal{A} + \mathcal{A}^\top \mathcal{X} & \mathcal{X} \mathcal{B}_p & \mathcal{X} \mathcal{B}_w & \mathcal{C}_z^\top \\
\mathcal{B}_p^\top \mathcal{X} & \Zero & \Zero & \mathcal{D}_{zp}^\top \\
\mathcal{B}_w^\top \mathcal{X} & \Zero & -\gamma \Ident & \mathcal{D}_{zw}^\top\\
\mathcal{C}_z & \mathcal{D}_{zp} & \mathcal{D}_{zw} & -\gamma \Ident
\end{pmatrix}
+
\begin{pmatrix} \ast \end{pmatrix}^\top
P
\begin{pmatrix}
\mathcal{C}_q & \mathcal{D}_{qp} & \mathcal{D}_{qw} & \Zero \\
\Zero & \Ident & \Zero & \Zero
\end{pmatrix}
\prec \Zero
\end{equation}
\begin{equation} \label{eq.iqc_lmi_modified} 
\begin{pmatrix}
\mathcal{X}\mathcal{A} + \mathcal{A}^\top \mathcal{X} & \mathcal{X} \mathcal{B}_p \Psi_2^{-1} & \mathcal{X} \mathcal{B}_w & \mathcal{C}_z^\top \\
(\Psi_2^{-1})^\top \mathcal{B}_p^\top \mathcal{X} & \Zero & \Zero & (\Psi_2^{-1})^\top \mathcal{D}_{zp}^\top \\
\mathcal{B}_w^\top \mathcal{X} & \Zero & -\gamma \Ident & \mathcal{D}_{zw}^\top\\
\mathcal{C}_z & \mathcal{D}_{zp} \Psi_2^{-1} & \mathcal{D}_{zw} & -\gamma \Ident
\end{pmatrix}
+
\begin{pmatrix} \ast \end{pmatrix}^\top
\hat{P}
\begin{pmatrix}
\Psi_1 \mathcal{C}_q & (\Psi_1 \mathcal{D}_{qp} + \Psi_3)\Psi_2^{-1} & \Psi_1 \mathcal{D}_{qw} & \Zero \\
\Zero & \Ident & \Zero & \Zero
\end{pmatrix}
\prec \Zero
\end{equation}
}
\end{figure*}
with decision variables $\gamma > 0$, $\mathcal{X}=\mathcal{X}^\top$ and the LMI variable arising from $\mathcal{P}$.
Minimizing $\gamma$ leads to robust performance level $\tilde{\gamma}_0$ and corresponding multiplier $P$.
Then consider the following factorization of the IQC multiplier,
\begin{equation} \label{eq.static_factorization}
P = 
\begin{pmatrix} \ast \end{pmatrix}^\top
\hat{P}
\begin{pmatrix}
	\Psi_1 & \Psi_3 \\
	\Zero_{q \times p} & \Psi_2
\end{pmatrix}
\end{equation}
with $\Psi_1 \in \R^{{n_q} \times {n_q}}$, $\Psi_3 \in \R^{{n_q} \times {n_p}}$, $\Psi_2 \in \R^{{n_p} \times {n_p}}$,
$\Psi_2$ invertible and where $\hat{P}$ has the form
\begin{equation}
\begin{pmatrix}
\Zero & \Ident \\
\Ident & \Zero
\end{pmatrix}
\text{ or }
\begin{pmatrix}
\Ident & \Zero\\
\Zero & - \Ident
\end{pmatrix}.
\end{equation}
Note that this factorization is possible for all static multiplier classes of interest.
Applying the factorization in \eqref{eq.iqc_lmi}, together with elementary manipulations and a congruence transformation
with $\diag(\Ident, \Psi_2^{-1}, \Ident, \Ident)$,  results in \eqref{eq.iqc_lmi_modified}.
We recognize that \eqref{eq.iqc_lmi_modified} is the initial analysis LMI \eqref{eq.iqc_lmi} for the plant $G_1 \starprod K_0$,
where
\begin{equation}
G_1 = 
\begin{bmatrix}
	\begin{array}{c | c  c  c}
	A 		& B_p \Psi_2^{-1}		& B_w 		& B_u \\
	\hline
	\Psi_1 C_q 	& \Psi_1D_{qp}\Psi_2^{-1} + \Psi_3\Psi_2^{-1}	& \Psi_1 D_{qw}	& \Psi_1 D_{qu} \\
	C_z 	& D_{zp}\Psi_2^{-1}	& D_{zw}	& D_{zu} \\
	C_y	& D_{yp}\Psi_2^{-1}	& D_{yw}	& D_{yu}
	\end{array}
\end{bmatrix}
\end{equation}
with new performance channels $w_1 \rightarrow z_1$ and $w \rightarrow z$.
We can now run a standard controller synthesis on the generalized plant $G_1$ with the quadratic performance criterion
\begin{equation} \label{eq.qp}
\left\langle \begin{pmatrix} z_1 \\ w_1 \end{pmatrix}, \hat{P} \begin{pmatrix} z_1 \\ w_1 \end{pmatrix} \right\rangle
+ \frac{1}{\gamma} \| w \|^2 - \gamma \| z \|^2 \leq - \epsilon ( \| w_1 \|^2 + \| w \|^2),
\end{equation}
e.g. along the lines described in \cite{schererweiland_lmi}. Minimization of $\gamma$ leads to the robustified controller $K_1$ and performance level $\gamma_1$. 
These steps can now be iterated until no substantial improvement in $\gamma$ is obtained.

Returning to the concrete setting of sector bounded nonlinearities, suppose that the method from Section \ref{sec.learning_static_nl} and \ref{sec.sector_bounds_from_gpr} resulted in sector bound estimates $[\hat{\kappa}_1^{(i)}, \hat{\kappa}_2^{(i)}]$, $i\in[n_p]$. To proceed with the synthesis, we need
\begin{assumption} \label{assumption.sc_sign}
For all $i\in [n_p]$ we have $\hat{\kappa}_1^{(i)} \leq 0 \leq \hat{\kappa}_2^{(i)}$.
\end{assumption}
We can now express the estimated sector bounds with full block multipliers.
\begin{lemma} \label{lemma.fb_multipliers}
Let $\delta \in (0,1)$ and run the learning method in Section \ref{sec.learning_static_nl} and \ref{sec.sector_bounds_from_gpr} on each nonlinearity using independent data sets and replacing $\delta$ with $\delta/{n_p}$. Define
\begin{equation} \label{eq.fb_multipliers}
\mathcal{P}_{\text{fb}} 
=
\left\{ P \mid 
	(\cdot)^\top P \begin{pmatrix} \Ident \\ \Theta(\theta) \end{pmatrix} \succeq \Zero,\:
	(\cdot)^\top P \begin{pmatrix} \Zero \\ \Ident \end{pmatrix} \preceq \Zero, \:
	\theta \in \hat{\mathcal{K}}
\right\},
\end{equation}
where $\hat{\mathcal{K}}=\left\{ (\theta_1,\ldots,\theta_{n_p}) \mid \theta_i \in \{ \hat{\kappa}_1^{(i)}, \hat{\kappa}_2^{(i)}\},\: i \in [n_p] \right\}$ and for brevity $\Theta(\theta)=\diag(\theta)$.
Then under Assumption \ref{assumption.noise}, \ref{assumption.rkhs}, \ref{assumption.sc_sign}, with probability at least $1-\delta$, \eqref{eq.iqc_def} holds for all $\Pi \in \mathcal{P}_{\text{fb}}$.
\end{lemma}
\begin{proof}
Follows immediately from \cite[Lemma~4.1]{veenmanscherer_iqc_synthesis}, Lemma \ref{lemma.sector_bounds_from_gpr} and the union bound.
\end{proof}
Since the multipliers from \eqref{eq.fb_multipliers} can be factorized according to \eqref{eq.static_factorization},
the synthesis-analysis procedure described above can be directly used.

\noindent
\emph{Remarks.} We would like to stress that the approach in \cite{veenmanscherer_iqc_synthesis} is much more general than the simplified setting used here for illustrative purposes. In particular, dynamic IQC multipliers can be used (involving non-trivial factorization results) which allow significantly more precise uncertainty descriptions leading to reduced conservatism. This is especially relevant for learning-based approaches, since the later might provide considerable additional information on the uncertainty. 
Furthermore, warm-start strategies are proposed in \cite{veenmanscherer_iqc_synthesis} in order to speed up the robust controller synthesis procedure. Finally, an approach for IQC-based robust synthesis in the context of gain scheduling is proposed in \cite{veenmanscherer_synthesis_robust_gainscheduling}. The extension of our approach to this setting is left for future work.

\subsection{From statistical to control-theoretic guarantees}
We can now give overall guarantees on the resulting controller.
\begin{theorem} \label{thm.controller_guarantees}
Let $\delta\in(0,1)$ and suppose that the learning procedure as outlined above is run with $\delta$ replaced by $\delta/n_p$ and $\mathcal{P}_{\text{fb}}$ is built according to \eqref{eq.fb_multipliers}. Under Assumption \ref{assumption.noise} to \ref{assumption.sc_sign}, if the controller synthesis algorithm described in Section \ref{sec.robust_synthesis} returns a controller $K$ after $M\geq 1$ iterations with robust performance level $\tilde{\gamma}_{M+1}$, then with probability at least $1-\delta$, $K$ will stabilize the true system and achieve robust performance level $\tilde{\gamma}_{M+1}$.
\end{theorem}
\begin{proof}
Follows immediate from Lemma \ref{lemma.fb_multipliers} and the results in \cite{veenmanscherer_iqc_synthesis}.
\end{proof}

Let us discuss this result. For the learning component, we assume that there exists a fixed (but unknown) ground truth, here the diagonal nonlinearity \eqref{eq.nl}. We derive from the learning component an uncertainty set, here described by the full block IQC multipliers \eqref{eq.fb_multipliers}, that contains the ground truth with given (high) probability $1-\delta$. This is combined with a robust method that comes with control-theoretic guarantees for this uncertainty set. Since the ground truth is contained in the uncertainty set with probability at least $1-\delta$, the guarantees hold with the same high probability. 

It is important to contrast this approach with randomized methods in robust control  \cite{tempocalafioredabbene_randomized}, like the scenario approach \cite{calafiorecampi_scenario_approach}. In these methods, high-probability guarantees are given with respect to the randomization in the algorithms, while in our problem the randomness comes only from the data.
Moreover, the results in this paper are also different in nature if compared to work done in a Bayesian framework. In particular, we do not rely on a prior distribution which might be misspecified. 
We argue that in the modern robust control setting, the frequentist perspective is the most natural statistical setting since we assume a fixed ground truth, described by an uncertainty class.
\section{EXAMPLE} \label{sec.example}
To demonstrate the advantage of the learning component in robust control, we now compare the controller synthesis using a-priori sector bounds and the bounds learned using GPR.
As a concrete example, we use the distillation column system from \cite[Section~8.2]{veenmanscherer_iqc_synthesis}, where we replace the original dead-zone nonlinearity with the unknown target function to be learned.
\subsubsection*{Setup}
The situation is depicted as a blockdiagram in Figure \ref{fig.distcol_blockdiagram}. The plant model
is given by
\begin{equation} \label{eq.distcol_plant}
G(s)=
\frac{1}{75s+1}
\begin{pmatrix}
87.8 & -86.4 \\
108.2 & -109.6
\end{pmatrix}
\end{equation}
and we use the performance weights
\begin{equation} \label{eq.weights}
W_e(s)=
\frac{s+0.1}{2s+10^{-5}}
\Ident_2
\qquad 
W_u(s)=
\frac{s+10}{s+100}
\Ident_2.
\end{equation}
The synthesis objective is therefore to track the reference signal $r$ at low frequencies
and penalizing control action at high frequencies.
\begin{figure}[h!]
\centering
\tikzstyle{block} = [draw, rectangle, 
    minimum height=3em, minimum width=3em]
\tikzstyle{sum} = [draw, circle, node distance=1cm]
\tikzstyle{input} = []
\tikzstyle{branch} = [coordinate]
\tikzstyle{signal} = []
\tikzstyle{output} = [coordinate]
\tikzstyle{pinstyle} = [pin edge={to-,thin,black}]

\resizebox{\columnwidth}{!}{
\begin{tikzpicture}[auto, node distance=2cm,>=latex',font=\huge]
	\node [input, name=input] {$r$};
	\node [sum, right of=input] (sumleft) {};
	
	\node [branch, right of=sumleft] (bwe) {}; 
	\node [block, above of=bwe] (We) {$W_e$};
	\node [signal, above of=We] (we) {$w_e$};
	
	\node [block, right of=bwe] (K) {$K$};
	
	\node [branch, right of=K] (bwu) {}; 
	\node [block, above of=bwu] (Wu) {$W_u$};
	\node [signal, above of=Wu] (wu) {$w_u$};
	
	\node [branch, right of=bwu] (bdelta) {};
	\node [coordinate, right of=bdelta] (pdelta) {};
	\node [block, above of=pdelta] (delta) {$\Delta$};
	
	\node[sum, right of=pdelta, node distance=2cm] (sumright) {};
	
	\node [block, right of=sumright] (plant) {$G$};
	\node[coordinate, below of=plant] (pplant) {};
	
	\node [branch, right of=plant] (by) {};
	\node [signal, right of=by] (y) {$y$};
	
	\draw [draw,->] (input)  -- node[pos=0.99] {$+$}(sumleft);
	\draw  (sumleft) -- node[pos=0.5] {$e$} (bwe);
	\draw [->] (bwe) -- (We);
	\draw [->] (We) -- (we);
	\draw [->] (bwe) -- (K);
	
	\draw (K) -- node[pos=0.5] {$u$} (bwu);
	\draw [->] (bwu) -- (Wu);
	\draw [->] (Wu) -- (wu);
	
	\draw (bwu) -- (bdelta) -- (pdelta);
	\draw [->] (bdelta) |- node[near end] {$q$} (delta);
	\draw [->] (pdelta) -- node[pos=0.99,auto=right] {$+$} (sumright);
	\draw [->] (delta) -| node[near start] {$p$} node[pos=0.99] {$-$} (sumright);
	
	\draw [->] (sumright) -- (plant);
	
	\draw (plant) -- (by);
	\draw [->] (by) -- (y);
	\draw (by) |- (pplant);
	\draw [->] (pplant) -| node[pos=0.99,auto=right] {$-$}(sumleft);
\end{tikzpicture} }
\caption{Block diagram of the example. Here $\Delta$ is a repeated static nonlinearity.
}
\label{fig.distcol_blockdiagram}
\end{figure}
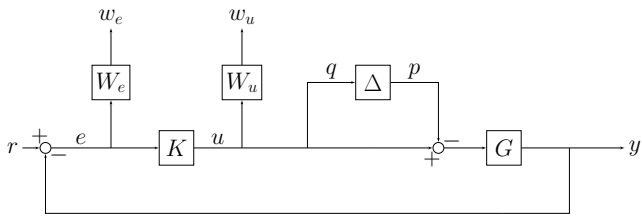
For illustrative purposes we use the same unknown nonlinearity in both uncertainty channels, i.e., $p_i=\phi(q_i)$, $i=1,2$.
For the experimental validation we need access to the ground truth, hence we use a function 
$\phi \in H_{k_0}$, where $k_0$ is given by \eqref{eq.kernel0} and for $k$ we choose the popular
Squared Exponential (SE) kernel $k(x,x^\prime)=\sigma_k^2\exp\left(-\frac{(x-x^\prime)^2}{2\ell^2} \right)$ with lengthscale $\ell=0.5$ and variance $\sigma_k^2=0.5$. The particular function $\phi$ is shown in Figure \ref{fig.sector_bounds} and has RKHS norm $ 2.6053$. Since we are interested in stabilization, we restrict attention to inputs to $\phi$ from $[-1,1]$ and $\phi$ fulfills the sector condition $[ -0.0572,0.3575]$ on this interval. 
Note that it is possible to ensure containment of the relevant signals in such an interval with techniques from robust control, but due to space constraints we omit the details.

Using standard manipulations, the robust synthesis problem with performance weights from \eqref{eq.weights} and diagonal static nonlinearities can be framed as an LFR. Furthermore, we use the full block multipliers from Lemma~\ref{lemma.fb_multipliers}. In the following experiments, we run the synthesis procedure outlined in Section \ref{sec.robust_synthesis} for 20 iterations.

Finally, to test the learning method we generate data sets $\mathcal{D}=((x_n,y_n))_{n\in[N]}$ by sampling $N=50$ inputs $x_n$ uniformely from $[-1,1]$, evaluate $\phi$ at these inputs and add independent $\Normal(0,0.05)$ noise to get $y_n$.

\subsubsection*{Improving performance with learning}
Assume we know a priori that $\phi$ belongs to the (rather wide) sector $[-0.9, 0.9]$. 
Running the robust controller synthesis leads to a controller with robust performance level $10.5575$.

This is now compared with the methodology described in Section \ref{sec.method}.
We use GPR with a zero prior mean function and the kernel $k_0$ as covariance function as well as the true noise level $0.05$.
We follow previous works, e.g. \cite{maddalenaetal_deterministic_uncertainty_bounds}, and assume an increased RKHS norm bound $B=2\|\phi\|_{k_0}$ in Proposition \ref{prop.uncertainty_bound}. 
Setting $\delta=0.001$ and following the procedure outlined in Section \ref{sec.sector_bounds_from_gpr} leads to the sector bounds $[-0.2460,  0.5480]$. 
The situation is illustrated in Figure \ref{fig.sector_bounds}.
Applying the controller synthesis method with this sector bound leads to a controller that has robust performance level  $3.1581$ at least with probability $1-0.001$.

\subsubsection*{Data-performance tradeoff} 
We now investigate the potential performance improvements with increasing size of the data set.
For this we generate data sets with 100, 150, 200, 250, 300 data points and applied the learning method to these data sets, using the same procedures as in the last experiment. We repeated this 50 times. The resulting estimates of the sector bounds are shown in Figure \ref{fig.data_performance_tradeoff} (upper two plots). It is clearly visible that more data is helpful for the learning method, since the sector bounds become tighter.
Furthermore, running the robust controller synthesis on the estimated sector bounds (averaged over all 50 trials) leads to the results shown in Figure \ref{fig.data_performance_tradeoff} (lower plot). Clearly, increased data helps to improve the robust performance. Note that the robust stability and performance of the controller is guaranteed with user-specified probability, while profiting from the additional information contained in the data set.

\begin{figure}
\centering
\includegraphics[width=0.4\textwidth]{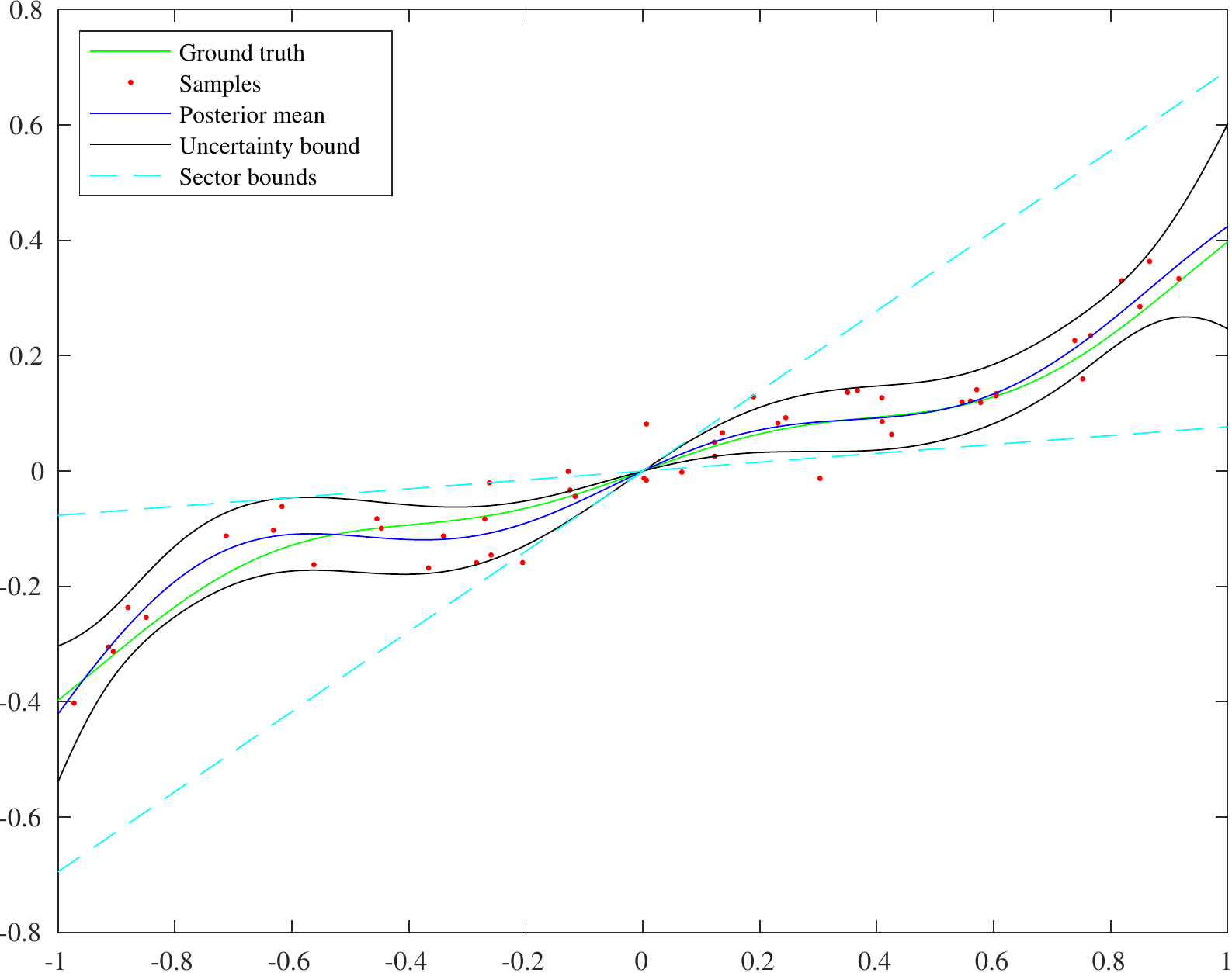}
\caption{The unknown nonlinearity is learned with GPR from samples, resulting in a high probability uncertainty set. This is used to derive sector bound estimates for the nonlinearity.}
\label{fig.sector_bounds}
\end{figure}
\begin{figure}
\centering
\includegraphics[width=0.4\textwidth]{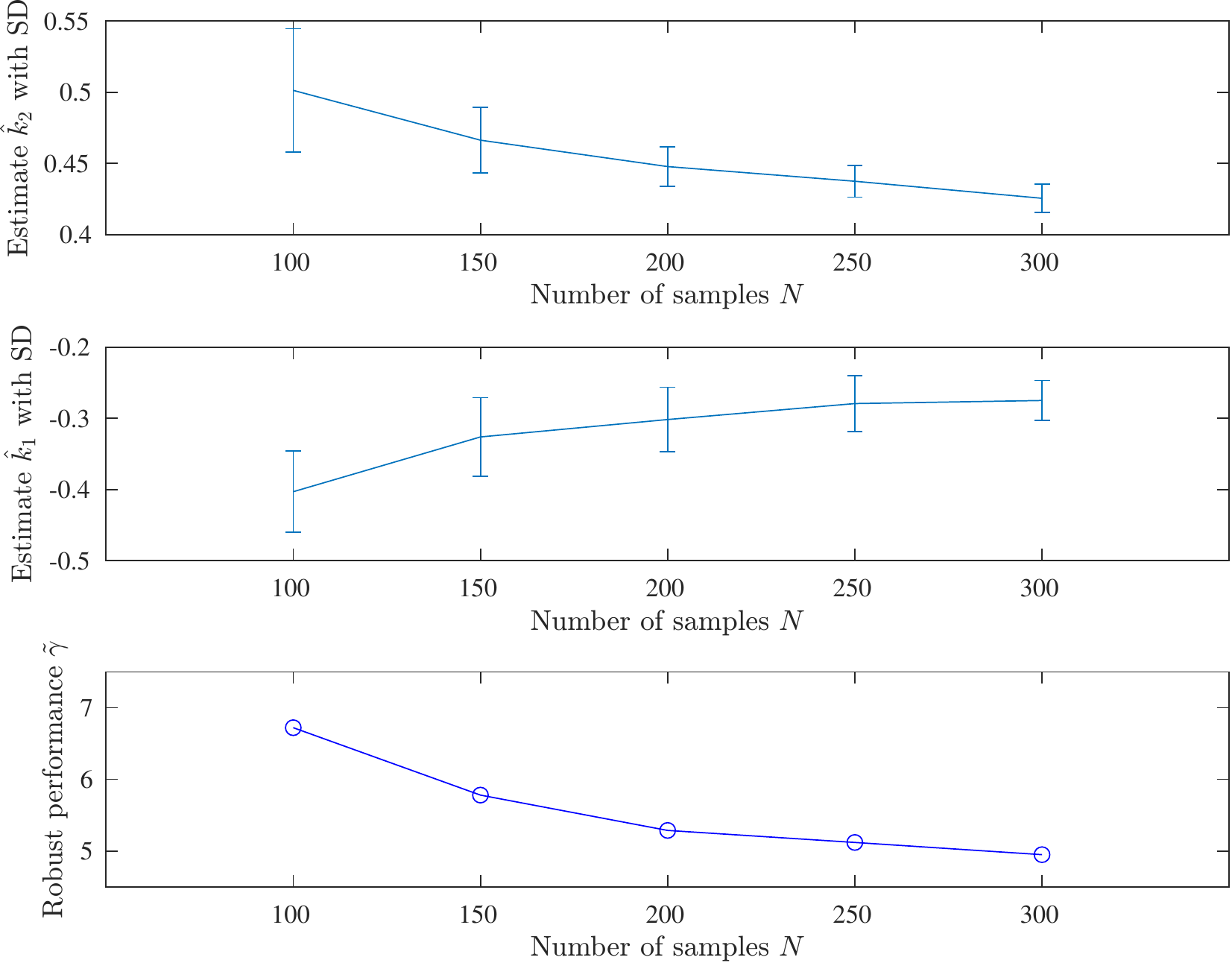}
\caption{Effect of increased data sets. In the upper two plots the estimated sector bounds for different data sizes are shown (averaged over 50 runs with corresponding empirical standard deviation), the lower plot shows the robust performance level of the resulting controllers.}
\label{fig.data_performance_tradeoff}
\end{figure}
\section{CONCLUSION} \label{sec.conclusion}
We presented a general framework integrating machine learning and modern robust control.
We showed how to utilize the flexible LFR and IQC frameworks to seamlessly include learning components in a robust control context. Furthermore, using GPR and recent statistical uncertainty results, we were able to give rigorous and practically meaningful statistical guarantees for the learning component that can be used in the IQC framework. In particular, we were able to use the learning components in a robust controller synthesis method and retain rigorous control-theoretic guarantees.
On a methodological level, we showed how to systematically include prior knowledge in the LFR setup and the learning component via customized kernels. Furthermore, we demonstrated how statistical guarantees without resorting to heuristics can be given for machine learning components and how to utilize these to arrive at combined statistical and control-theoretic guarantees.
A concrete example demonstrated the practical feasibility of the overall approach.

As already indicated, the general approach offers a lot of flexibility. Interesting directions for future work include the integration of several learning components, additional uncertainties without learning, and other synthesis settings like robust gain scheduling.
Ongoing work is concerned with the extension of results like Proposition \ref{prop.uncertainty_bound} to linear system identification with kernel methods as well as using different learning methods, e.g. to deal with non-stochastic, but bounded noise assumptions.

\section*{ACKNOWLEDGMENTS}
We thank Joost Veenman for providing his extensive code and simulation files from \cite{veenmanscherer_iqc_synthesis}, and Steve Heim and Alexander von Rohr for their feedback on a draft of this paper.

                                  

\bibliography{main}

\end{document}